\documentclass[submission]{eptcs}
\providecommand{\event}{GandALF 2012} 
\usepackage{url}
\usepackage{amsmath}	
\usepackage{amsthm}		
\usepackage{bussproofs}
\usepackage{array}
\usepackage{proof}
\usepackage[utf8]{inputenc}

\newcommand{\uks}{\ensuremath{\mathclose{{}^{*}}}}

\newcommand{\eql}{=}		
\newcommand{\eqd}{=}	

\newcommand{\ob}{\{\,}
\newcommand{\cb}{\,\}}
\newcommand{\st}{\;|\;}		
\newcommand{\ccup}{\:\cup\:}

\newcommand{\f}[1]{\mathsf{#1}}
\newcommand{\del}{\Delta_{\alpha p}}
\newcommand{\delw}{\hat{\Delta}_w}
\newcommand{\e}{\mathsf{E}_\alpha}
\newcommand{\D}{\mathsf{D}_{\alpha p}}
\newcommand{\lang}{\mathcal{L}}
\newcommand{\hd}{\mathsf{hd}}
\newcommand{\der}{\mathsf{der}_{\alpha p}}

\newtheorem{lemma}{Lemma}

\newtheorem{proposition}{Proposition}
\newtheorem{corollary}{Corollary}

\title{Deciding KAT and Hoare Logic with Derivatives%
\thanks{This work was partially funded by the European Regional Development
Fund through the programme COMPETE and by the Portuguese Government
through the FCT – Fundação para a Ciência e a Tecnologia under the
project PEst-C/MAT/UI0144/2011 and by project
CANTE-PTDC/EIA-CCO/101904/2008.}
}
\author{
\makebox[0cm][c]{%
Ricardo Almeida}
\institute{ 
Departamento de Ci\^encia de Computadores\\ 
Faculdade de Ci\^encias, Universidade do Porto}
\email{up030308017@alunos.dcc.fc.up.pt}
\and
\makebox[0cm][c]{%
Sabine Broda}
\hspace{40mm}%
\makebox[0cm][c]{%
Nelma Moreira}%
\institute{Centro de Matem\'atica da Universidade do Porto\\
Departamento de Ci\^encia de Computadores,\\ 
Faculdade de Ci\^encias, Universidade do Porto}
\email{
 \makebox[0cm][c]{sbb@dcc.fc.up.pt}\hspace{40mm}%
       \makebox[0cm][c]{nam@dcc.fc.up.pt}
}
}

\def\titlerunning{Deciding KAT and Hoare Logic with Derivatives}
\def\authorrunning{R. Almeida, S. Broda \& N. Moreira}

\begin{document}
\maketitle

\begin{abstract}
Kleene algebra with tests (KAT) is an equational system for program
verification, which is the combination of Boolean algebra (BA) and
Kleene algebra (KA), the algebra of regular expressions.  In
particular, KAT subsumes the propositional fragment of Hoare logic
(PHL) which is a formal system for the specification and verification
of programs, and that is currently the base of most tools for checking
program correctness. Both the equational theory of KAT and the
encoding of PHL in KAT are known to be decidable.  In this paper we
present a new decision procedure for the equivalence of two KAT
expressions based on the notion of partial derivatives. We also
introduce the notion of derivative modulo particular sets of
equations.  With this we extend the previous procedure for deciding
PHL. Some experimental results are also presented.
 \end{abstract}

\section{Introduction}
Kleene algebra with tests (KAT) is an equational algebraic system for
reasoning about programs that combines Kleene algebra (KA) with
Boolean algebra~\cite{kozen96:_kleen_algeb_with_tests}. In particular,
KAT subsumes
PHL~\cite{kozen00:_hoare_logic_and_kleen_algeb_with_tests}, the
propositional fragment of Hoare logic, which is a formal system for the
specification and verification of programs, and that is currently the
base of most tools for checking program
correctness~\cite{hoare69:_axiom_basis_for_comput_progr}.  Testing if
two KAT expressions are equivalent is tantamount to prove that two
programs are equivalent or that a Hoare triple is valid. Deciding the
equivalence of KAT expressions is as hard as deciding regular
expressions (KA expression) equivalence,
i.e. PSPACE-complete~\cite{cohen96:_compl_of_kleen_algeb_with_tests}. In
spite of KAT's success in dealing with several software verification
tasks, there are very few software applications that implement KAT's
equational theory and/or provide adequate decision procedures. Most
of them are within (interactive) theorem provers or part of model
checking systems,
see~\cite{aboul-hosn06:_kat_ml,hofner07:_autom_reason_in_kleen_algeb,babu11:_chop_expres_and_discr_durat_calcul}
for some examples. 

Based on a rewrite system of Antimirov and
Mosses~\cite{antimirov94:_rewrit_exten_regul_expres}, Almeida \emph{et
  al.}~\cite{almeida09:_antim_mosses} developed an algorithm that
decides regular expression equivalence through an iterated process of
testing the equivalence of their derivatives, without resorting to the
classic method of minimal automaton comparison. Statistically
significant experimental tests showed that this method is, on average
and using an uniform distribution,
more efficient than the classical methods based on
automata~\cite{almeida11:_equiv_of_regul_languag}. Another advantage of
this method is that it is easily adapted to other Kleene algebra, such
as KAT. In this paper we present an extension of that decision
algorithm to test equivalence in KAT. The termination and correctness
of the algorithm follow the lines of~\cite{almeida09:_antim_mosses},
but are also close to the coalgebraic approach to KAT presented by
Kozen~\cite{kozen08:_coalg_theor_of_kleen_algeb_with_tests}.  Deciding
PHL can be reduced to testing KAT expressions
equivalence~\cite{kozen00:_hoare_logic_and_kleen_algeb_with_tests}. Here
we present an alternative method by extending the notion of derivative
modulo a set of (atomic equational) assumptions. Once again the
decision procedure has to be only slightly adapted. The new method
reduces the size of the KAT expressions to be compared with the cost
of a preprocessing phase. All the procedures were implemented in OCaml
and some experimental results are also presented.

\section{Preliminaries}
\label{sec:preliminares}
We briefly review some basic definitions about regular
expressions, Kleene algebras, Kleene algebras with tests (KAT), and
KAT expressions. For more details, we refer the reader
to~\cite{kozen94:_compl_theor_for_kleen_algeb,kozen97:_kleen_algeb_with_tests,kozen96:_kleen_algeb_with_tests,kozen03:_autom_guard_strin_and_applic,cohen96:_compl_of_kleen_algeb_with_tests}.
\subsection{Kleene Algebra and Regular Expressions}
\label{sec:ka}
Let $\Sigma=\{p_1,\ldots,p_k \}$, with $k\geq 1$, be an \emph{alphabet}.
A \emph{word} $w$ over $\Sigma$
is any finite sequence of letters. The \emph{empty word} is denoted by
$1$.  Let $\Sigma\uks$ be the set of all words over
$\Sigma$. A \emph{language} over $\Sigma$ is a subset of
$\Sigma\uks$.
The \emph{left quotient} of a language $L\subseteq \Sigma\uks$ by a
word $w\in \Sigma\uks$ is the language
$w^{-1}L=\{x\in\Sigma\uks\mid wx \in L\}$.
The set of \emph{regular expressions} over $\Sigma$, $R_\Sigma$, is defined
by:
\begin{eqnarray}
  r&:=&0 \mid 1 \mid p \in \Sigma \mid (r_1 + r_2) \mid (r_1 \cdot r_2) 
  \mid r\uks
\end{eqnarray}
where the operator $\cdot$ (concatenation) is often omitted.  The
language $\lang(r)$ associated to $r$ is inductively defined
as follows: $\lang(0)=\emptyset$,
$\lang(1)=\{1\}$, $\lang(p)=\{p\}$ for
$p \in \Sigma$,
$\lang(r_1+r_2)=\lang(r_1)\cup\lang(r_2)$,
$\lang(r_1\cdot r_2)=\lang(r_1)\cdot\lang(r_2)$, and
$\lang(r\uks)=\lang(r)\uks$.    
Two regular expressions $r_1$ and $r_2$
are \emph{equivalent} if $\lang(r_1)=\lang(r_2)$, and we write
$r_1 =  r_2$. With this interpretation, the algebraic structure
$(R_\Sigma,+,\cdot,0,1)$ constitutes an idempotent
semiring, and with the unary operator $\uks$, a Kleene algebra.

A Kleene algebra is an algebraic structure $\mathcal{K} = (K, +,
\cdot, \uks, 0, 1)$, satisfying the axioms below.
\begin{align} \label{ka:1}	
r_1 + (r_2 + r_3) &=  (r_1 + r_2) + r_3\\ \label{ka:2}	 
       r_1 + r_2 &=  r_2 + r_1  \\
 \label{ka:3}       r + 0 &=  r+r \; = r \\
\label{ka:4}        r_1(r_2r_3)  &=  (r_1r_2)r_3  \\
\label{ka:5}	  1r  &=  r1 \; =r	 \\
    r_1(r_2 + r_3) &=  r_1r_2 + r_1r_3  \label{ka:6}\\
    (r_1 + r_2)r_3  &=  r_1r_3 + r_2r_3  \label{ka:7}\\
	  0r  &=  r0 \; =0 	 \label{ka:8} \\
    1 + rr\uks &\le  r\uks\label{ka:9} \\
    1 + r\uks r  &\le r\uks \label{ka:10}\\
    r_1 + r_2r_3 \le r_3 &\rightarrow r_2\uks r_1 \le r_3 \label{ka:11}\\
    r_1 + r_2r_3 \le r_2 &\rightarrow r_1r_3\uks \le r_2 \label{ka:12}
  \end{align}

  In the above, $\le$ is defined by $r_1 \le r_2$ if and only if $r_1 +
  r_2 = r_2$.  The axioms say that the structure is an
  \emph{idempotent semiring} under $+$, $\cdot$, $0$ and $1$ and that
  $\uks$ behaves like the Kleene star operator of formal language
  theory. This axiom set (with an usual first-order deduction system)
  constitutes a complete proof system for equivalence between regular
  expressions~\cite{kozen94:_compl_theor_for_kleen_algeb}.

\subsection{Kleene Algebra with Tests and KAT Expressions}
\label{sec:kat}
A Kleene algebra with tests (KAT) is a Kleene algebra with an embedded
Boolean subalgebra $\mathcal{K} = (K, B, +, \cdot, \uks, 0, 1, \bar{
})$ where $\ \bar{}\ $ is an unary operator denoting negation and is defined only on $B$, such that
\begin{itemize}
 \item $(K, +, \cdot, \uks, 0, 1)$ is a Kleene algebra;
 \item $(B, +, \cdot, \bar{}\:, 0, 1)$ is a Boolean algebra;
 \item $(B, +, \cdot, 0, 1)$ is a subalgebra of $(K, +, \cdot, 0,
   1)$.
\end{itemize}
\noindent Thus, a KAT is an algebraic structure that satisfies the KA
axioms (\ref{ka:1})--(\ref{ka:12}) and  the axioms for a Boolean algebra $B$.
 
Let $\Sigma=\{p_1,\ldots,p_k \}$ be a non-empty set of (primitive)
\emph{action} symbols and $T=\{t_1,\ldots,t_l \}$ be a non-empty set
of (primitive) \emph{test} symbols.  The set of boolean expressions
over $T$ is denoted by $\f{Bexp}$ and the set of KAT expressions by
$\f{Exp}$, with elements $b_1$, $b_2, \ldots$ and $e_1$,
$e_2, \ldots$, respectively. The abstract syntax of KAT expressions over an alphabet
$\Sigma \ccup T$ is given by the following grammar,
\begin{eqnarray*}
  b \in \f{Bexp}&:=& 0 \mid 1 \mid t \in T \mid \overline{b} \mid b_1 + b_2
  \mid b_1 \cdot b_2 \\
  e \in \f{Exp}&:=&p \in \Sigma \mid b \in \f{Bexp} \mid e_1 + e_2 \mid e_1 \cdot e_2 
  \mid e_1 \uks .
\end{eqnarray*}
As usual, we often omit the operator $\cdot$ in concatenations and in
conjunctions. The standard language-theoretic models of KAT are regular sets of
\emph{guarded strings} over alphabets $\Sigma$ and
$T$~\cite{kozen03:_autom_guard_strin_and_applic}. Let
$\overline{T}\eqd\{\overline{t}\mid t \in T\}$
and let $\f{At}$ be
the set of \emph{atoms}, i.e., of all truth assignments to T,
\begin{align*}
  \f{At} &= \{ b_1 \ldots b_l \mid b_i \text{ is either $t_i$ 
    or   $\overline{t_i}$ for $1 \leq i \leq l$ and $t_i \in T$}\}.
\end{align*} 
Then the set of \emph{guarded strings} over $\Sigma$ and $T$ is
$\f{GS} = (\f{At} \cdot \Sigma)\uks \cdot \f{At}$. Guarded
strings will be denoted by $x,y,\ldots$.  For
$x=\alpha_1p_1\alpha_2p_2 \cdots p_{n-1}\alpha_n \in \f{GS}$, where $n
\geq 1$, $\alpha_i \in \f{At}$ and $p_i \in \Sigma$, we define
$\f{first}(x)=\alpha_1$ and
$\f{last}(x)=\alpha_n$. If $\f{last}(x) =
\f{first}(y)$, then the \emph{fusion product} $xy$ is defined by
concatenating $x$ and $y$, omitting the extra occurrence of the common
atom. If $\f{last}(x) \neq \f{first}(y)$, then $xy$ does not
exist.
For sets $X, Y \subseteq \f{GS}$ of guarded strings, the set $X
\diamond Y$ defines the set of all $xy$ such that $x
\in X$ and $y \in Y$. We have that $X^0 \eqd\f{At}$ and $X^{n+1} \eqd X
\diamond X^n$, for $n\geq 0$.

Every KAT expression $e\in \f{Exp}$ denotes a set of \emph{guarded
  strings}, $\f{GS}(e) \subseteq \f{GS}$. Given a KAT expression $e$
we define $\f{GS}(e)$ inductively as follows,
\begin{equation*}
  \label{eq:gs}
\begin{array}{llllll}
  \f{GS}(p) &  \eqd  & \ob \alpha_1 p \alpha_2 \mid \alpha_1,\alpha_2
  \in \f{At} \cb & \qquad p \in \Sigma\\
  \f{GS}(b) &  \eqd  & \ob \alpha \in \f{At} \st \alpha \leq b \cb
  & \qquad b \in \f{Bexp} \\
  \f{GS}(e_1+e_2) &  \eqd  & \f{GS}(e_1) \ccup \f{GS}(e_2) 	&  \\
  \f{GS}(e_1e_2) &  \eqd  & \f{GS}(e_1) \diamond \f{GS}(e_2) 	&  \\
  \f{GS}(e\uks) &  \eqd  & \cup_{n \geq 0} \f{GS}(e)^n .	& 
\end{array}
\end{equation*}
We say that two KAT expressions $e_1$ and $e_2$ are \emph{equivalent},
and write $e_1 = e_2$, if and only if $\f{GS}(e_1) \eql \f{GS}(e_2)$.
Kozen~\cite{kozen96:_kleen_algeb_with_tests} showed that one has $e_1
=e_2$ modulo the KAT axioms, if and only if, $e_1 = e_2$ is true in
the free Kleene algebra with tests on generators $\Sigma \cup T$.  Two
sets of KAT expressions $E,F\subseteq \f{Exp}$ are \emph{equivalent}
if and only if $\f{GS}(E) \eql \f{GS}(F)$, where $\f{GS}(E) = \cup_{e
  \in E} \f{GS}(e)$.

\section{Deciding Equivalence in KAT}
\label{equivKAT}

In this section we present a decision algorithm to test equivalence in
KAT. Kozen~\cite{kozen08:_coalg_theor_of_kleen_algeb_with_tests}
presented a coalgebraic theory for KAT extending Rutten's coalgebraic
approach for KA~\cite{rutten98:_autom_and_coind_exerc_in_coalg}, and
improving the framework of Chen and
Pucella~\cite{chen04:_coalg_approac_to_kleen_algeb_with_tests}.
Extending the notion of Brzozowski derivatives to KAT, Kozen
proved the existence of a coinductive equivalence procedure. Our
approach follows closely that work, but we explicitly define the
notion of partial derivatives for KAT, and we effectively provide a
(inductive) decision procedure. This decision procedure is an
extension of the algorithm for deciding equivalence of regular
expressions given
in~\cite{almeida09:_antim_mosses,antimirov94:_rewrit_exten_regul_expres},
that does not use the axiomatic system. Equivalence of expressions is
decided through an iterated process of testing the equivalence of
their partial derivatives.

\subsection{Derivatives}
\label{sec:derivatives}
Given a set of guarded strings $R$, its derivative with respect to
$\alpha p \in \f{At} \cdot \Sigma$, denoted by $\f{D}_{\alpha p}(R)$,
is defined as being the \emph{left quotient} of $R$ by $\alpha p$.  As such,
one considers the following \emph{derivative} functions,
\begin{equation*}
\begin{array}{lllll}
  \f{D}:  \f{At} \cdot \Sigma  \rightarrow  {\cal
    P}(\f{GS}) \rightarrow {\cal P}(\f{GS}) &  \ \ \ \ \ \  \  \ \ \f{E} :\f{At}  \rightarrow  {\cal P}(\f{GS}) \rightarrow \{0,1\}
\end{array}
 \end{equation*}
\noindent consisting of components,
\begin{equation*}
\begin{array}{lllll}
  \f{D}_{\alpha p} : {\cal P}(\f{GS}) \rightarrow {\cal P}(\f{GS}) & \ \ \ \ \ \ \ \ \ \ \ \ \f{E}_{\alpha} : {\cal P}(\f{GS}) \rightarrow \{0,1\}	\end{array}$$
defined as follows. For $\alpha \in \f{At}$, $p \in \Sigma$ and $R
\subseteq \f{GS}$,
$$\begin{array}{lll}
 \f{D}_{\alpha p} (R) \:\: &\eqd & \ob y \in \f{GS} \st \alpha p y \in R \cb 	
\end{array} \qquad \mbox{ and } \qquad
\begin{array}{lll}
 \f{E}_\alpha (R) \:\: &\eqd \:\:& \left\{\,  \begin{array}{ll} 
					   1   & \text{if } \alpha \in R \\
					   0  & \text{otherwise.}
					  \end{array} \right.
\end{array}
\end{equation*}

\vspace{-0.5cm}

\subsection{Partial Derivatives}
\label{sec:partialderivatives}
The notion of set of \emph{partial derivatives},
cf.~\cite{antimirov96:_partial_deriv_regul_expres_finit_autom_const,mirkin66},
corresponds to a finite set representation of the derivatives of an
expression.  Given  $\alpha \in \f{At}$, $p \in \Sigma$ and $e\in
\f{Exp}$, the set $\del(e)$ of partial derivatives of $e$ with respect
to $\alpha p$ is inductively defined as follows,
\begin{align*}
\f{\Delta} : \f{At} \cdot \Sigma  &\:\rightarrow\: \f{Exp}\:
\rightarrow\: {\cal P}(\f{Exp}) \\
 \del (p')  &= \left\{\begin{array}{ll}	
					   \{1\}   & \text{if } p = p' \\
					   \emptyset   & \text{otherwise}
					  \end{array} \right.		 \\
 \del (b)  &=  \emptyset	\\
 \del (e_1 + e_2)  &=  \del(e_1) \ccup \del (e_2)		\\
 \del (e_1 e_2)  &= \left\{\begin{array}{ll}	
					   \del(e_1) \cdot e_2  \  & \text{if } \f{E}_\alpha(e_1)=0 \\
					   \del(e_1) \cdot e_2 \ccup \del(e_2)   & \text{if } \f{E}_\alpha(e_1)=1
					  \end{array} \right.\\
 \del (e\uks)  &=  \del (e) \cdot e\uks,
\end{align*}
where for $\Gamma\subseteq \f{Exp}$ and $e\in \f{Exp}$,
 $\Gamma \cdot e \eql \ob e'e \st e' \in \Gamma \cb$ if $e \neq 0$ and
 $e \neq 1$, and $\Gamma \cdot 0 \eql \emptyset$ and $\Gamma \cdot 1
 \eql \Gamma$, otherwise. We note that $\del(e)$ corresponds to an
 equivalence class of $D_{\alpha p}(e)$ (the syntactic Brzozowski
 derivative, defined
 in~\cite{kozen08:_coalg_theor_of_kleen_algeb_with_tests}) modulo
 axioms (\ref{ka:1})--(\ref{ka:3}), (\ref{ka:5}), (\ref{ka:7}), and
 (\ref{ka:8}). Kozen calls such a structure a \emph{right
   presemiring}.

The following syntactic definition of $\f{E}_\alpha: \f{At}   \:\rightarrow\: \f{Exp} \rightarrow\: \{0,1\}$ is
from~\cite{kozen08:_coalg_theor_of_kleen_algeb_with_tests} and 
simply evaluates an expression with respect to the truth assignment $\alpha$.
\vspace{-0.3cm}


$$\begin{array}{ll}
 \f{E}_\alpha (p)  &=  0 \\
 \f{E}_\alpha (b)  &=  \left\{\begin{array}{ll}	
					   1   & \text{if } \alpha \leq b \\
					   0   & \text{otherwise}
					  \end{array} \right.	\\
\end{array} \qquad
\begin{array}{ll}
 \f{E_{\alpha}} (e_1 + e_2)  &=  \f{E_{\alpha}}(e_1) + \f{E_{\alpha}} (e_2)	\\
 \f{E_{\alpha}} (e_1 e_2)  &=  \f{E_{\alpha}}(e_1) \f{E_\alpha}(e_2)	\\
 \f{E_{\alpha}} (e\uks)  &= 1.		
\end{array}$$


One can show that,
\begin{align*}
\f{E_{\alpha}} (e)  &= \left\{\begin{array}{ll}	
					   1   & \text{if } \alpha \leq e \\
					   0  & \text{if } \alpha \not\leq e
					  \end{array} \right.  =  \left\{\begin{array}{ll}	
					   1   & \text{if } \alpha \in \f{GS}(e) \\
					   0   & \text{if } \alpha \notin \f{GS}(e).
					  \end{array} \right.
                                    \end{align*}

The next proposition shows that for all KAT expressions $e$ the set of
guarded strings correspondent to the set of partial derivatives of $e$
w.r.t. $\alpha p\in \f{At}\cdot \Sigma$ is the derivative of
$\f{GS}(e)$ by $\alpha p$.
\theoremstyle{plain}
\begin{proposition}\label{prop:partialderiv}
 For all KAT expressions $e$, all atoms $\alpha$ and all symbols $p$, 
  $$\D(\f{GS}(e)) \eql \f{GS}(\del(e)).$$ 
\end{proposition}
\begin{proof}
  The proof is obtained by induction on the structure of $e$.  We
  exemplify with the case $e=e_1e_2$, where
	  $$\begin{array}{llll}
	    \D(\f{GS}(e)) &\eql& \D(\f{GS}(e_1) \diamond \f{GS}(e_2)) \\
	    &\eql& \left\{\begin{array}{lr}	
					   \D(\f{GS}(e_1)) \diamond
                                           \f{GS}(e_2)& \ \ \ \ \ \ \
                                           \ \ \ \ \ \ \text{ if } \alpha \notin \f{GS}(e_1) \\
					   \D(\f{GS}(e_1)) \diamond \f{GS}(e_2) \ccup \D(\f{GS}(e_2))  & \text{if } \alpha \in \f{GS}(e_1)
					  \end{array} \right.  		 \\
	  && \mbox{applying the induction hypothesis}		 \\
	  &\eql& \left\{\begin{array}{lr}	
					   (\cup_{e' \in \del(e_1)}
                                           \f{GS}(e')) \diamond
                                           \f{GS}(e_2)   & \ \ \ \ \ \
                                           \text{if } \f{E}_\alpha(e_1)=0 \\
					   (\cup_{e' \in \del(e_1)} \f{GS}(e')) \diamond \f{GS}(e_2) \ccup \f{GS}(\del(e_2))   & \text{if } \f{E}_\alpha(e_1)=1
					  \end{array} \right.   \\
	  &\eql& \left\{\begin{array}{lr}	
              \cup_{e' \in \del(e_1)}
                                           \f{GS}(e' e_2)    & \ \ \ \
                                           \ \ \ \ \ \ \ \ \ \ \ 
                                           \ \ \text{if } \f{E}_\alpha(e_1)=0 \\
					   (\cup_{e' \in \del(e_1)} \f{GS}(e'e_2))  \ccup \f{GS}(\del(e_2))  & \text{if } \f{E}_\alpha(e_1)=1
					  \end{array} \right.   \\
	  &\eql& \left\{\begin{array}{lr}	
                                           \f{GS}(\del(e_1)\cdot e_2)   & \ \ \ \ \ \ \ \ \ \ \ \
                                           \ \ \ \ \ \ \ \ \ \ \ \ \ \
                                         \text{if } \f{E}_\alpha(e_1)=0 \\
					   \f{GS}(\del(e_1)\cdot e_2)  \ccup \f{GS}(\del(e_2))  & \text{if } \f{E}_\alpha(e_1)=1
					  \end{array} \right.  \\
	  &\eql& \f{GS}(\del(e_1e_2)) \eql \f{GS}(\del(e)) .
	  \end{array}$$
\end{proof}

The notion of partial derivative of an expression w.r.t. $\alpha p \in
\f{At} \cdot \Sigma$ can be extended to words  $x \in (\f{At} \cdot \Sigma)\uks$, as follows,
\begin{align*}
 \f{\hat{\Delta}} : (\f{At} \cdot \Sigma)\uks  \:\rightarrow\: \f{Exp}
 &\rightarrow\: {\cal P}(\f{Exp}) \\
 \f{\hat{\Delta}}_1 (e) &= \{e\}		\\
 \f{\hat{\Delta}}_{w \alpha p}(e)  &= \del (\hat{\Delta}_w(e)) .
\end{align*}
Here, the notion of (partial) derivatives has been extended to sets of KAT
expressions $E\subseteq \f{Exp}$, by defining, as expected,
$\f{\Delta}_{\alpha p}(E)=\cup_{e\in E}\f{\Delta}_{\alpha p}(e)$, for
$\alpha p \in \f{At}\cdot\Sigma$. Analogously, we also consider
$\f{\hat{\Delta}}_x(E)$ and $\f{\hat{\Delta}}_R(E)$, for $x\in
(\f{At} \cdot \Sigma)\uks$ and $R\subseteq (\f{At} \cdot \Sigma)\uks$.

The fact, that for any $e\in \f{Exp}$ the set
$\f{\hat{\Delta}}_{(\f{At} \cdot \Sigma)\uks}(e)$ is finite, ensures the
termination of the decision procedure presented in the next section.
\subsection{A Decision Procedure for  KAT Expressions Equivalence} 
\label{secKAT:equiv}
In this section we describe an algorithm for testing the equivalence
of a pair of KAT expressions using partial derivatives.  Following
Antimirov~\cite{antimirov96:_partial_deriv_regul_expres_finit_autom_const},
and for the sake of efficiency, we define the function $\f{f}$ that
given an expression $e$ computes the set of pairs $(\alpha p, e')$, such that for each $\alpha
p\in \f{At} \cdot \Sigma$, the corresponding $e'$ is a partial
derivative of $e$ with respect to $\alpha p$.
\begin{align*}
  \f{f}: \: \f{Exp} \:\: &\rightarrow \:\: {\cal P} (\f{At}\cdot  \Sigma
  \times \f{Exp})		\\
  \f{f}(p)  &=   \ob(\alpha p, 1) \;\; | \;\; \alpha \in \f{At} \cb	\\
  \f{f}(b)  &=  \emptyset	\\
  \f{f}(e_1 + e_2)  &=  \f{f}(e_1) \ccup \f{f}(e_2) 	\\
  \f{f}(e_1 e_2) &=  \f{f}(e_1)\cdot e_2 \ccup \ob (\alpha p, e) \in \f{f}(e_2) \; | \; \f{E_\alpha}(e_1)=1 \cb	 \\
  \f{f}(e\uks) &= \f{f}(e) \cdot e\uks \
\end{align*}
where, as before, $\Gamma \cdot e = \ob (\alpha p, e'e) \st (\alpha p
,e') \in \Gamma \cb$ if $e \neq 0$ and $e \neq 1$, and $\Gamma \cdot 0
\eql \emptyset$ and $\Gamma \cdot 1 \eql \Gamma$, otherwise.  Also, we
denote by $\hd(\f{f}(e))=\{ \alpha p \mid (\alpha p, e') \in
\f{f}(e)\}$ the set of \emph{heads} (i.e. first components of each
element) of $\f{f}(e)$. The function $\der$, defined in
(\ref{def:der}), collects all the partial derivatives of an expression
$e$ w.r.t. $\alpha p$, that were computed by function $\f{f}$.
\begin{align}
\label{def:der}
 \f{der}_{\alpha p}(e)  = \{e' \mid (\alpha p,e') \in \f{f}(e)\}
\end{align}
The proof of the following Proposition is almost trivial and follows
from the symmetry of the definitions of $\der$, $\f{f}$, and $\del$.
\theoremstyle{plain}
\begin{proposition}\label{prop:derpd}
  For all $e,e' \in \f{Exp}$, $\alpha \in \f{At}$ and $p \in \Sigma$
 one has, $ \der(e) \eql \del(e). $
\end{proposition}
\bigskip
To define the decision procedure we need to consider the above
functions and the ones defined in
Section~\ref{sec:partialderivatives} applied to sets of KAT
expressions. 
Then, we define the function $\f{derivatives}$ that given two
sets of KAT expressions $E_1$ and $E_2$ computes all pairs of sets of
partial derivatives of $E_1$ and $E_2$ w.r.t. $\alpha p \in \f{At}
\cdot \Sigma$, respectively.
\begin{align*}
 \f{derivatives}: {\cal P}(\f{Exp})^2 &\rightarrow \: {\cal P} ({\cal P}(\f{Exp})^2)\\
 \f{derivatives}(E_1,E_2)\:  &=  \{(\der(E_1),\der(E_2))\mid \alpha  p \in \hd(E_1 \ccup E_2)\}
\end{align*}
Finally, we present the function $\f{equiv}$ that tests if two (sets
of) KAT expressions are equivalent. For two sets of KAT expressions
$E_1$ and $E_2$ the function returns $\f{True}$, if for every atom
$\alpha$, $\f{E_\alpha}(E_1) = \f{E_\alpha}(E_2)$ and if, for
every $\alpha p$, the partial derivative of $E_1$ w.r.t.~$\alpha p$ is
equivalent to the partial derivative of $E_2$ w.r.t. $\alpha p$.
\begin{align*}
 \f{equiv}:\:{\cal P}({\cal P}(\f{Exp})^2)  \times  {\cal P}({\cal P}(\f{Exp})^2) &\rightarrow \{\f{True},\f{False}\}   \\
 \f{equiv}(\emptyset,H)  &=  \f{True}       \\
 \f{equiv}(\{(E_1,E_2)\} \ccup S, H)  &=    \left\{
             \begin{array}{lr}
               \f{False} & \text{if } \exists \alpha \in \f{At} : \f{E_{\alpha}}(E_1) \neq \f{E_{\alpha}}(E_2) \\
               \f{equiv}(S \ccup S', H') & \text{otherwise,}
             \end{array} \right.       
        \end{align*}
      
where
\begin{align*}
   S' = \{ d \mid d \in
\f{derivatives}(E_1,E_2) \text{ and }  d \notin H' \} & 
\text{\ \  and\ \  }  H' = \{(E_1,E_2) \} \ccup H.
\end{align*}
The function $\f{equiv}$ accepts two sets $S$ and $H$ as arguments. At
each step,  
 $S$ contains the pairs of (sets of) expressions that still need to be
checked for equivalence, whereas $H$ contains the pairs of (sets of) expressions
that have already been tested.  The use of the set $H$ is important to
ensure that the derivatives of the same pair of (sets of) expressions are not
computed more than once, and thus prevent a possible infinite
loop. 

To compare two expressions $e_1$ and $e_2$, the initial call must be
$\f{equiv}(\{(\{e_1\},\{e_2\})\},\emptyset)$. At each step the
function takes a pair $(E_1,E_2)$ and verifies if there exists an atom
$\alpha$ such that $\f{E_\alpha}(E_1) \neq \f{E_\alpha}(E_2)$. If such an
atom exists, then $e_1 \not= e_2$ and the function halts, returning
$\f{False}$. If no such atom exists, then the function adds
$(E_1,E_2)$ to $H$ and then replaces in $S$ the pair $(E_1,E_2)$ by
the pairs of its corresponding derivatives provided that these are not
in $H$ already. The return value of $\f{equiv}$ will be the result of
recursively calling $\f{equiv}$ with the new sets as arguments.  If
the function ever receives $\emptyset$ as $S$, then the initial call
ensures that $e_1 =e_2$, since all derivatives have been successfully
tested, and the function returns $\f{True}$.

\subsection{Termination and Correctness}
\label{sec:correctness}
First, we show that the function $\f{equiv}$ terminates. For  every KAT
expression $e$, we define the set $\f{PD}(e)$  and show that, for every KAT
expression $e$, the set of partial derivatives of $e$ is a subset  of
$\f{PD}(e)$, which on the other hand is clearly finite. The set
$\f{PD}(e)$  coincides with the \emph{closure} of a KAT expression $e$,
defined by Kozen, and is also similar to Mirkin's prebases~\cite{mirkin66}.
$$\begin{array}{ll}
  \f{PD}(b) &= \{b\} \\
  \f{PD}(p) &= \{p,1\}	\\
\end{array} \qquad
\begin{array}{ll}
  \f{PD}(e_1+e_2) &= \{e_1+e_2\} \ccup \f{PD}(e_1) \ccup \f{PD}(e_2)\\
  \f{PD}(e_1e_2) &= \{e_1e_2\} \ccup \f{PD}(e_1) \cdot e_2 \ccup \f{PD}(e_2)\\
  \f{PD}(e\uks) &= \{e\uks\} \ccup \f{PD}(e) \cdot e\uks.
\end{array}$$

\begin{lemma}	
\label{lemma:lemmaTerminating}
Consider $e, e' \in \f{Exp}$, $\alpha \in \f{At}$ and
 $p \in \Sigma$. If $e' \in \f{PD}(e)$, then 
$\del(e') \subseteq \f{PD}(e)$.
\end{lemma}
\begin{proof}
 The proof is obtained by induction on the structure of $e$. We exemplify with the case $e \eql e_1e_2$.
 Let $e' \in \f{PD}(e_1e_2) \:\eql\: \{e_1e_2\} \ccup \f{PD}(e_1)
 \cdot e_2 \ccup \f{PD}(e_2)$.
\begin{itemize}
\item If $e' \in \{e_1e_2\}$, then $\del(e') \subseteq \del(e_1) \cdot
  e_2 \ccup \del(e_2)$. But $e_1 \in \f{PD}(e_1)$ and $e_2 \in
  \f{PD}(e_2)$, so applying the induction hypothesis twice, we obtain
  $\del(e') \subseteq \f{PD}(e_1) \cdot e_2 \ccup \f{PD}(e_2)
  \subseteq \f{PD}(e).$
\item If $e' \in \f{PD}(e_1) \cdot e_2$, then $e'=e'_1e_2$ such that
  $e'_1 \in \f{PD}(e_1)$. So $\del(e') \subseteq \del(e'_1) \cdot e_2
  \ccup \del(e_2)$ $\subseteq \f{PD}(e_1) \cdot e_2 \ccup \f{PD}(e_2)
  \subseteq \f{PD}(e)$.
 \item Finally, if $e' \in \f{PD}(e_2)$, again by the induction
   hypothesis we have $\del(e') \subseteq \f{PD} (e_2) \subseteq
   \f{PD}(e)$.
\end{itemize}
\end{proof}

\begin{proposition}\label{prop:finitenessParDerivatives}
  For all $x \in (\f{At} \cdot \Sigma)\uks$, one has
  $\hat{\Delta}_x(e) \subseteq \f{PD}(e)$.
\end{proposition}

\begin{proof}
  We prove this lemma by induction on the length of $x$.  If $|x|=0$,
  i.e.~$x=1$, then $\hat{\Delta}_1(e) \:\eql\: \{e\} \:\subseteq\:
  \f{PD}(e)$.  If $x=w \alpha p$, then $\hat{\Delta}_{w \alpha p} \eql
  \cup_{e' \in \hat{\Delta}_w(e)} \del(e')$. By induction hypothesis,
  we know that $\hat{\Delta}_w(e) \subseteq \f{PD}(e)$. By
  Lemma~\ref{lemma:lemmaTerminating}, if $e' \in \f{PD}(e)$, then
  $\del(e') \subseteq \f{PD}(e)$. Consequently, $\cup_{e' \in
    \delw(e)} \del(e') \subseteq \f{PD}(e)$.
\end{proof}

\begin{corollary}\label{corolary} For all KAT expressions $e$, the
  set $\hat{\Delta}_{(\f{At} \cdot \Sigma)\uks} (e)$ is finite.
\end{corollary}
It is obvious that the previous results also apply to sets of KAT expressions.

\begin{proposition}\label{prop:termination}
The function $\f{equiv}$ is terminating.  
\end{proposition}
\begin{proof}
  When the set $S$ is empty it follows directly from the definition of
  the function that it terminates.  We argue that when $S$ is not
  empty the function also terminates based on these two aspects:
\begin{itemize}
\item In order to ensure that the set of  partial derivatives of a pair of (sets of) expressions
  are not computed more than once, the set $H$ is used to store the
  ones which have already been calculated.
\item Each function call removes one pair $(E_1,E_2)$ from the set $S$
  and appends the set of partial derivatives of $(E_1,E_2)$, which
  have not been calculated yet, to $S$. By  {Corollary
    \ref{corolary}}, the set of partial derivatives of an expression
  by any  word is finite, and so eventually
  $S$ becomes $\emptyset$.
\end{itemize}
Thus, since at each call the function analyzes one pair from $S$,
after a finite number of calls the function terminates.
\end{proof}

The next proposition states the correctness of our
algorithm. Coalgebraically it states that two KAT expressions are
equivalent if and only if there exists a bisimulation between
them~\cite[Thm. 5.3]{kozen08:_coalg_theor_of_kleen_algeb_with_tests}.

\begin{proposition}\label{prop:correctnessequiv}
For all KAT expressions $e_1$ and $e_2$,
\[ \f{GS}(e_1) \eql \f{GS}(e_2) \qquad \Leftrightarrow \qquad \left\{
             \begin{array}{ll}
               \f{E}_\alpha (e_1) \eql \f{E}_\alpha (e_2)  \qquad \text{ and}&   \\
               \f{GS}(\del(e_1)) \eql  \f{GS}(\del(e_2)), & \quad \forall {\alpha \in \f{At}, \;\forall p \in \Sigma }.
             \end{array} \right. \]  
\end{proposition}

\begin{proof}
 Let us first prove the $\Leftarrow$ implication. If $\f{GS}(e_1) \neq
 \f{GS}(e_2)$, then there is $x \in \f{GS}$, such that $x \in
 \f{GS}(e_1)$ and $x \notin \f{GS}(e_2)$ (or vice-versa).
If $x=\alpha$, then  we have $\e(e_1) =1 \neq 0= \e(e_2)$
    and the test fails.
If $x=\alpha p w$, such that $w \in (\f{At} \cdot \Sigma)\uks \cdot \f{At}$, then since
    $\alpha p w \in \f{GS}(e_1)$ and $\alpha p w \notin \f{GS}(e_2)$, we have that
    $ w \in \f{GS}(\del(e_1))$ and  $w \notin \f{GS}(\del(e_2)).$
    Thus, $\f{GS}(\del(e_1)) \neq \f{GS}(\del(e_2)).$

    Let us now prove the $\Rightarrow$ implication.  For $\alpha \in
    \f{At}$, there is either $\alpha \in \f{GS}(e_1)$ and $\alpha \in
    \f{GS}(e_2)$, thus $\f{E}_\alpha (e_1) \eql \f{E}_\alpha (e_2)=1$;
    or $\alpha \not\in \f{GS}(e_1)$ and $\alpha \not\in \f{GS}(e_2)$,
    thus $\f{E}_\alpha (e_1) \eql \f{E}_\alpha (e_2)=0$. For
    $\alpha p \in \f{At} \cdot \Sigma$, by
    Proposition~\ref{prop:partialderiv}, one has $\f{GS}(\del(e_1)) \eql \f{GS}(\del(e_2))$
    if and only if $\D(\f{GS}(e_1)) = \D(\f{GS}(e_2))$. This follows
    trivially from $\f{GS}(e_1) \eql \f{GS}(e_2)$.
\end{proof}

\section{Implementation}
\label{sec:implementation}
The algorithm presented in the previous section was implemented in
\emph{OCaml}~\cite{team12:_ocaml}. Alternations, conjunctions, and
disjunctions are represented by sets, and thus, commutativity and
idempotence properties are naturally enforced. Concatenations are
represented by lists of expressions. Primitive tests occurring in a KAT
expression are represented by integers, and atoms by lists of boolean
values (where primitive tests correspond to indexes). For each KAT
expression $e$, we consider $\f{At}$ as the set of atoms that correspond
to the primitive tests that occur in $e$.  The implementation of the
functions defined in Section~\ref{sec:partialderivatives} and
Section~\ref{secKAT:equiv}, do not differ much from their formal
definitions. A common choice was the use of comprehension lists to
define the inclusion criteria of elements in a set. Because of our
basic representation of KAT expressions, we treat in a uniform way
both expressions and sets of expressions. The function $\f{E}_\alpha$,
used in $\f{equiv}$, is implemented using a function called
$\f{eAll}$, that takes as arguments two (sets of) expressions $E_1$
and $E_2$ and verifies if for every atom the truth assignments for
$E_1$ and $E_2$ coincide.

\subsection{Experimental Results}
\label{sec:experimental}
In order to test the performance of our decision procedure we ran some
experiments. We used the \textsf{FAdo} system~\cite{fado_online} to
uniformly random generate samples of KAT expressions. Each sample has
10000 KAT expressions of a given length $|e|$ (number of symbols in
the syntactic tree of $e\in \f{Exp}$). The size of each sample is more
than enough to ensure results statistically significant with 95\%
confidence level within a 5\% error margin.  The tests were executed
in the same computer, an Intel\textsuperscript{\textregistered}
Xeon\textsuperscript{\textregistered} 5140 at 2.33\, GHz with 4\, GB
of RAM, running a minimal 64 bit Linux system.  For each sample we
performed two experiments: (1) we tested the equivalence of each KAT
expression against itself; (2) we tested the equivalence of two
consecutive KAT expressions.  For each pair of KAT expressions we
measured: the size of the set $H$ produced by $\f{equiv}$ (that
measures the number of iterations) and the number of primitive tests in
each expression ($|e|_T$). Table~\ref{tab:exp:random} summarizes some of the results obtained.
Each row corresponds to a sample, where the three first columns
characterize the sample, respectively, the number of primitive actions
($k$), the number of primitive tests ($l$), and the length of each KAT
expression generated.
Column four has the number of primitive tests in
each expression ($|e|_T$).
  Columns five and six give the average size of
$H$ in the experiment (1) and (2), respectively. Column seven is the
ratio of the equivalent pairs in experiment~(2).
Finally, columns eight and nine contain
the average times, in seconds, of each comparison in the experiments (1) and (2). 
More than comparing with existent systems, which is difficult by the
reasons pointed out in the introduction, these experiments aimed to test the feasibility
of the procedure. As expected, the main \emph{bottleneck} is the
number of different primitive tests in the KAT expressions.

\begin{table}[ht]
  \centering
  \begin{small}
   \begin{tabular}{|c|c|c||c|c|c|c|c|c|c|c|}\hline
1& 2& 3 & 4  & 5& 6  & 7 & 8 & 9 \\
\hline
 $k$& $l$& $|e|$ & $|e|_T$  & $H$(1)& $H$(2)  & $=$(2) &Time(1) &Time(2) \\
  \hline\hline
$5$&$5$&$50$     &$9.98$    & $7.35$ &$0.53$   &$0.42$ &$0.0097$&$0.00087$ \\ \hline
$5$&$5$&$100$    &$19.71$   &$15.74$ &$0.76$   &$0.48$ &$0.0875$&$0.00223$\\\hline\hline
$10$&$10$&$50$   &$11.12$   &$8.30$  &$0.50$   &$0.07$ &$0.5050$&$0.30963$\\\hline
$10$&$10$&$100$  &$21.93$   &$16.78$ &$0.67$   &$0.18$ &$20.45$ &$1.31263$\\\hline\hline
$15$&$15$&$50$   &$11.57$   &$8.47$  &$0.47$   &$0.10$ &$6.4578$&$55.22$\\\hline           
  \end{tabular}
\end{small}
  \caption{Experimental results for uniformly random generated KAT expressions.}
  \label{tab:exp:random}
\end{table}


\section{Hoare Logic and KAT}
\label{sec:HL}
\emph{Hoare logic} was first introduced in 1969,
cf.~\cite{hoare69:_axiom_basis_for_comput_progr}, and is a formal
system widely used for the specification and verification of
programs. Hoare logic uses \emph{partial correctness assertions}
(PCA's) to reason about program correctness. A PCA is a triple, $\{b\}
 P  \{c\}$ with $P$ being a program, and $b$ and $c$ logic
formulas. We read such an assertion as \emph{if $b$ holds before the
  execution of $P$, then $c$ will necessarily hold at the end of the
  execution, provided that $P$ halts}.  A deductive system of Hoare
logic provides inference rules for deriving valid PCA's, where rules
depend on the program constructs. We consider a simple \textbf{while}
language, where a program $P$ can be defined, as usual, by
an assignment $x:=v$; a $\mathbf{skip}$ command; a sequence $ P; Q$,
conditional $ \mathbf{if}\; b\;
\mathbf{then}\; P \;\mathbf{else} \; Q$, and a loop $ \mathbf{while} \;b\;\mathbf{do}\; P$.

There are several variations of Hoare logic and here we choose an
inference system, considered
in~\cite{frade11:_verif_condit_for_sourc_level_imper_progr}, that
enjoys the \emph{sub-formula} property,  where the premises of a
rule can be obtained from the assertions that 
occur in the rule's conclusion. With this property, given a PCA
$\{b\}P\{c\}$, where $P$ has also some annotated assertions, it is possible
to automatically generate verification conditions that will ensure
its validity. The inference rules for this system are the following:\\

\AxiomC{$b \to c$}
\UnaryInfC{$\{b\} \; {\bf skip} \; \{c\}$}
\DisplayProof \hspace{0.5cm}
\AxiomC{$b \to c[x/e]$}
      \UnaryInfC{$\{b\}\; x:= e \;\{c\}$}
\DisplayProof\hspace{0.5cm}
       \AxiomC{$\{b\} \; P \; \{c\}$}
       \AxiomC{$\{c\} \; Q \; \{d\}$}
\BinaryInfC{$\{b\} \; P \; ;\{c\} \; Q \; \{d\}$}
\DisplayProof

\medskip
      \AxiomC{$\{b \land c\} \; P \; \{d\}$}
       \AxiomC{$\{\neg b \land c\} \; Q \; \{d\}$}
\BinaryInfC{$\{c\} \; \mathbf{if}\; b\; \mathbf{then}\; P \;\mathbf{else}\; Q \; \{d\}$}
\DisplayProof \hspace{1cm}
              \AxiomC{$\{b \land i\} \; P \; \{i\}$}
              \AxiomC{$ c \to i$}
              \AxiomC{$(i \wedge \neg b) \to d$}
           \TrinaryInfC{$\{c\} \; \mathbf{while}\; b\; \mathbf{do}\; \{i\}
                P\; \{d\}$}
\DisplayProof






\subsection{Encoding Propositional Hoare Logic in KAT} 
\label{sec:modelling}
The propositional fragment of Hoare logic (PHL), i.e., the fragment
without the rule for assignment, can be encoded in
KAT~\cite{kozen00:_hoare_logic_and_kleen_algeb_with_tests}.  The
encoding of an annotated \textbf{while} program $P$ and of our
inference system follow the same lines. In PHL, all assignment
instructions are represented by primitive symbols $p$. The
$\mathbf{skip}$ command is encoded by a distinguished primitive symbol
$p_\mathbf{skip}$. If $e_1$, $e_2$ are respectively the encodings of
programs $P_1$ and $P_2$, then the encoding of more complex constructs
of an annotated \textbf{while} program involving $P_1$ and $P_2$ is as
follows.
\begin{align*}
	 P_1 \; ; \{c\} \; P_2 \:\:  &\Rightarrow \:\:  e_1ce_2		\\
        \mbox{\textbf{if} $b$ \textbf{then} $P_1$ \textbf{else}
            $P_2$}  \:\: &\Rightarrow \:\:  be_1 + \bar{b}e_2 	\\ 
       \textbf{while} \; b \; \textbf{do} \; \{i\}\;P_1 \:\:&\Rightarrow \:\:  (bie_1)\uks \bar{b}
     \end{align*}
   
A PCA of the form $\{b\} P  \{c\}$ is encoded in KAT as an equational identity of the form
\begin{equation}\label{ru:assign}
 b e = b e c \ \ \ \ \ \ \ \ \ \mbox{ or equivalently by} \ \ \ \ \ \
 \ \ \ be\overline{c}=0, \nonumber
\end{equation}
where $e$ is the encoding of the program $P$.

Now, suppose we want to prove the PCA $\{b\} P \{c\}$. Since the
inference system for Hoare logic, that we are considering in this
paper, enjoys the sub-formula property, one can generate mechanically
in a backward fashion the verification conditions that ensure the PCA's
validity.

Since in the KAT encoding,  $be\overline{c}=0$, we do not have the
rule for assignment, besides  verification conditions
(proof obligations) of the form $b' \rightarrow c'$ we will also have
assumptions of the form  $b'p\overline{c'}=0$.

One can generate a set of
assumptions, $\Gamma=\f{Gen}(be\overline{c})$,
backwards from $be\overline{c}=0$, where $\f{Gen}$ is inductively
defined by:

$$\begin{array}{llll}
\f{Gen}(b \; p_{\bf skip} \; \overline{c}) & = & \{ b \leq c \} \\
\f{Gen}(b \; p \; \overline{c}) & = & \{ b \; p \; \overline{c} \}
&p_{\bf skip} \not= p \in \Sigma\\
\f{Gen}(b \; e_1 \; c \; e_2 \;\overline{d}) & = & \f{Gen}(b \; e_1
\; \overline{c}) \cup \f{Gen}(c \; e_2 \;\overline{d}) \\
\f{Gen}(b \; (ce_1 + \bar{c}e_2) \; \overline{d}) &=& \f{Gen}(bc \; e_1
\; \overline{d}) \cup \f{Gen}(b\overline{c} \; e_2 \;\overline{d}) \\ 
\f{Gen}(b \; ((cie)\uks \bar{c}) \; \overline{d}) &=& \f{Gen}(ic \; e
\; \overline{i}) \cup \{ b \leq i, i\overline{c} \leq d \}  
\end{array}$$

Note that $\Gamma$ is necessarily of the form
$$\Gamma = \{ b_1p_1\overline{b_1'}=0, \ldots ,\ b_mp_m\overline{b_m'}=0\} \cup
\{c_1\leq c_1', \ldots , c_n \leq c_n' \},$$
where $p_1,\ldots,p_m \in \Sigma$ and such that all $b$'s and $c$'s
are $\f{Bexp}$ expressions.
In Section~\ref{sec:decidingHoare}, we show how one can prove the validity of
$be_P\overline{c}=0$ in the presence of such a set of assumptions $\Gamma$,
but first we illustrate the encoding and generation of the assumption set with an example.

\subsection{A Small Example}
\label{subsec:ex}
Consider the program P in Table~\ref{tab:prog}, that calculates the
factorial of a non-negative integer. We wish to prove that, at the end
of the execution, the variable $y$ contains the factorial of $x$, i.e.~to verify the assertion
$\{\f{True}\}\; \mbox{P} \; \{y=x!\}$.

\begin{table}[ht]
  \centering
  \begin{small}
    \begin{tabular}{|l|l|c|}\hline
Program $P$
&Annotated Program $P'$
&Symbols used  \\
&& in the encoding
\\
\hline
 & $y := 1;$ & $p_1$	\\
& $\{y = 0!\}$& $t_1$	\\
$y := 1$;& $z := 0;$& $p_2$	\\
$z := 0;$& $\{y = z!\}$& $t_2$	\\
while $\neg z=x$ do& while $\neg z=x$ do& $t_3$	\\
\{& \{& \\
\mbox{\quad z := z+1;}& \mbox{\quad \{y=z!\}}& $t_2$	\\
\mbox{\quad y := y$\times$z;}& \mbox{\quad z := z+1;}& $p_3$	\\
\mbox{\}}& \mbox{\quad \{y$\times$z = z!\}}& $t_4$	\\
& \mbox{\quad y := y$\times$z;}	 &$p_4$	 \\
& \mbox{\}}& \\
\hline
  \end{tabular}
\end{small}
\medskip
  \caption{A program for the factorial}
  \label{tab:prog}
\end{table}
In order to apply the inference rules we need to annotate program $P$,
obtaining program $P'$. Applying the inference rules for deriving
PCA's in a backward fashion to $\{\f{True}\}\; \mbox{P}' \; \{y=x!\}$,
one easily generates the corresponding set of assumptions provided by
the annotated version of the program. However, because we do not have
the assignment rule in the KAT encoding, here we simulate that by
considering not only verification conditions but also atomic PCA's
$\{b'\} x:=e \{c'\}$.  Thus the assumption set is
\begin{align*}
\Gamma_P &= \left\{ \begin{array}{l}
\{\f{True}\} y:=1 \{y=0!\},\{\f{y=0!}\} z:=0  \{y=z!\}, \\
\{\f{y=z!\land\neg z=x}\} z := z+1 \{y \times z = z!\}, \{ y \times z
= z!\} y := y \times z  \{y=z!\}, \\
y = z! \rightarrow y = z!,  (y = z! \wedge \neg \neg z=x)
\rightarrow y=x! \;
\end{array}\right\}.
\end{align*}

On the other hand, using the correspondence of KAT primitive symbols and atomic
parts of the annotated program $P'$, as in Table~\ref{tab:prog},
and additionally encoding $\f{True}$ as $t_0$ and $y=x!$ as $t_5$,
respectively, the encoding of  $\{\f{True}\}\;
\mbox{P}' \; \{y=x!\}$ in KAT is
\begin{equation}
t_0p_1t_1p_2t_2(t_3t_2p_3t_4p_4)\uks \overline{t_3}\overline{t_5} = 0.
\label{eq:pca}	
\end{equation}
The corresponding set of assumptions $\Gamma$ in KAT is
\begin{equation}
\Gamma = \{ t_0p_1\overline{t_1}=0, t_1p_2\overline{t_2}=0,
t_2t_3p_3\overline{t_4}=0, t_4p_4\overline{t_2}=0, t_2 \leq t_2,
t_2\overline{t_3} \leq t_5 \}. \label{assumptions}	
\end{equation}
In the next section we will see how to prove in KAT an equation such
as~(\ref{eq:pca}) in the presence of a set of assumptions such as
(\ref{assumptions}).

\section{Deciding Hoare Logic}
\label{sec:decidingHoare}
Rephrasing the observation in the end of last section, we are
interested in proving in KAT the validity of implications of the form
\begin{equation}\label{eq:assumptions}
  b_1p_1\overline{b_1'}=0 \ \land  \cdots \land \ b_mp_m\overline{b_m'}=0
\ \land \ c_1 \leq c_1' \ \land \cdots \land c_n \leq c_n' \ \ \to \ \
bp\overline{b'}=0.
\end{equation}
This can be reduced to proving the equivalence of
KAT expressions, since it has been shown,
cf.~\cite{kozen00:_hoare_logic_and_kleen_algeb_with_tests}, that for
all KAT expressions $r_1, \ldots, r_n,e_1,e_2$ over $\Sigma = \{ p_1 ,
\ldots , p_k \}$ and $T = \{ t_1 , \ldots , t_l \}$, an implication of
the form
$$ r_1 =0 \ \land \cdots \land \ r_n=0 \;\; \to \;\; e_1 = e_2$$
is a theorem of KAT if and only if 
\begin{equation}
  \label{eq:uru}
  e_1 + uru = e_2 + uru
\end{equation}
where $u= (p_1 + \cdots + p_k)\uks$ and $r = r_1 + \ldots +
r_n$. Testing this last equality
can of course be done by applying our algorithm to $e_1 + uru$ and $e_2 + uru$. However, in
the next subsection, we present an alternative method of proving the
validity of implications of the form \ref{eq:assumptions}. This method
has the advantage of prescinding from the expressions $u$ and $r$, above. 

\subsection{Equivalence of KAT Expressions Modulo a Set of Assumptions}
\label{sec:equivwithaxioms}

In the presence of a finite set of assumptions of the form 
\begin{equation}
  \label{eq:gamma}
\Gamma = \{b_1p_1\overline{b_1'}=0,\ldots,b_mp_m\overline{b_m'}=0\} \cup \{c_1\leq c_1', \ldots , c_n \leq c_n'\}
\end{equation}
we have to restrict ourselves to atoms that satisfy the restrictions
in $\Gamma$. Thus, let 
\begin{equation}
\f{At}^{\Gamma}=\{\; \alpha \in \f{At} \mid  \alpha \leq c \to
\alpha \leq c', \text{ for all } c\leq c' \in
\Gamma \;\}.
\end{equation}

Given a KAT expression $e$, the \emph{set of guarded strings modulo}
$\Gamma$, $\f{GS}^{\Gamma}(e)$, is inductively defined as follows.
$$\begin{array}{llllll}
 \f{GS}^{\Gamma}(p) &  \eqd  & \ob \alpha p \beta \mid
 \alpha, \beta \in \f{At}^{\Gamma}  \wedge
 \forall_{bp\overline{b'}=0 \;\in \; \Gamma} \ ( \alpha \leq b \to
 \beta \leq b') \cb & \\
 \f{GS}^{\Gamma}(b) &  \eqd  & \ob \alpha \in \f{At}^{\Gamma} \st \alpha \leq b \cb  \\
 \f{GS}^{\Gamma}(e_1+e_2) &  \eqd  & \f{GS}^{\Gamma}(e_1) \ccup \f{GS}^{\Gamma}(e_2) 	&  \\
 \f{GS}^{\Gamma}(e_1e_2) &  \eqd  & \f{GS}^{\Gamma}(e_1) \diamond \f{GS}^{\Gamma}(e_2) 	&  \\
 \f{GS}^{\Gamma}(e\uks) &  \eqd  & \cup_{n \geq 0} \f{GS}^{\Gamma}(e)^n .	& 
\end{array}$$  

The following proposition characterizes the \emph{equivalence}
modulo a set of assumptions $\Gamma$, and ensures the correctness of
the new Hoare logic decision procedure.

\medskip
\theoremstyle{plain}
\begin{proposition}
\label{prop:gammacorrectness}
Let $e_1$ and $e_2$ be  KAT expressions and $\Gamma$ a set of
assumptions as in (\ref{eq:gamma}). Then,
$$\f{KAT},\Gamma \vdash e_1 = e_2 \qquad \text{ iff } \qquad
\f{GS}^{\Gamma}(e_1) = \f{GS}^{\Gamma}(e_2).$$  
\end{proposition}
\begin{proof}
  By~(\ref{eq:uru}) one has $\f{KAT},\Gamma \vdash e_1 = e_2$ if and only if $e_1 +
  uru = e_2 + uru$ is provable in KAT, where $u= (p_1 + \cdots +
  p_k)\uks$ and $r = b_1p_1\overline{b_1'} + \cdots +
  b_mp_m\overline{b_m'} + c_1 \overline{c_1'} + \cdots + c_n
  \overline{c_n'}$. The second equality is equivalent to
  $\f{GS}(e_1 + uru)=\f{GS}(e_2 + uru)$, i.e.~$\f{GS}(e_1) \cup
  \f{GS}(uru) = \f{GS}(e_2) \cup \f{GS}(uru)$.  In order to show the
  equivalence of this last equality and $\f{GS}^{\Gamma}(e_1) =
  \f{GS}^{\Gamma}(e_2)$, it is sufficient to show that for every KAT
  expression $e$ one has $\f{GS}^{\Gamma}(e)= \f{GS}(e) \setminus
  \f{GS}(uru)$ (note that $A \cup C = B \cup C \ \Leftrightarrow \ A
  \setminus C = B \setminus C$).

  First we analyze under which conditions a guarded string $x$ is an
  element of $\f{GS}(uru)$. Given the values of $u$ and $r$, it is
  easy to see that $x \in \f{GS}(uru)$ if and only if in $x$ occurs an
  atom $\alpha$ such that $\alpha \leq c$ and $\alpha \not\leq c'$ for
  some $c \leq c' \in \Gamma$, or $x$ has a substring $\alpha p
  \beta$, such that $\alpha \leq b$ and $\alpha \not\leq b'$ for some
  $bp\overline{b'} \in \Gamma$.  This means that $x \not\in
  \f{GS}(uru)$ if and only if every atom in $x$ is an element of
  $\f{At}^{\Gamma}$ and every substring $\alpha p \beta$ of $x$
  satisfies $(\alpha  \leq b \to \beta \leq b')$, for all  ${bp\overline{b'}=0 \;\in \; \Gamma}.$
From this remark and by the definitions of $\f{At}^{\Gamma}$ and
$\f{GS}^{\Gamma}$, we conclude that $\f{GS}^{\Gamma}(e) \cap
\f{GS}(uru) = \emptyset$. 
Note also that, since $\f{GS}^{\Gamma}(e)$ is a restriction of
$\f{GS}(e)$, one has $\f{GS}^{\Gamma}(e) \subseteq
\f{GS}(e)$. Now it suffices to show that for every $x \in \f{GS}(e)
\setminus \f{GS}(uru)$, one has $x \in \f{GS}^{\Gamma}(e)$. This can
be easily proved by induction on the structure of $e$. 
\end{proof}

We now define the set of partial derivatives of a KAT expression
modulo a set of assumptions $\Gamma$. Let $e\in \f{Exp}$.  If $\alpha
\not \in \f{At}^{\Gamma}$, then $\del^{\Gamma} (e) = \emptyset$.  For
$\alpha \in \f{At}^{\Gamma}$, let
\begin{align*}
 \del^{\Gamma} (p')  &= \left\{\begin{array}{ll}	
                                           \{\Pi b' \mid  bp\overline{b'}=0 \in
                                           \Gamma \land \alpha \leq b \;\} & \text{if }
                                           p=p' \\
					   \emptyset   & \text{if } p
                                           \not= p'
					  \end{array} \right.	 \\
 \del^{\Gamma} (b)  &=  \emptyset	\\
 \del^{\Gamma} (e_1 + e_2)  &=  \del^{\Gamma}(e_1) \ccup \del^{\Gamma} (e_2)		\\
 \del^{\Gamma} (e_1 e_2)  &= \left\{\begin{array}{lr}	
					   \del^{\Gamma}(e_1) \cdot e_2  \  & \text{if } \f{E}_\alpha(e_1)=0 \\
					   \del^{\Gamma}(e_1) \cdot
                                           e_2 \ccup
                                           \del^{\Gamma}(e_2)   \qquad
                                           & \qquad \text{if } \f{E}_\alpha(e_1)=1
					  \end{array} \right.\\
 \del^{\Gamma} (e\uks)  &=  \del^{\Gamma} (e) \cdot e\uks.
\end{align*}
Note, that by definition, $\Pi \; b' = 1$ if there is no $bp=bpb' \in
\Gamma$ such that $\alpha \leq b$ and $\alpha \in \f{At}^{\Gamma}$.
The next proposition states the
correctness of the definition of $\del^{\Gamma}$.
\theoremstyle{plain}
\begin{proposition}
 Let $\Gamma$ be a set of assumptions as above, $e\in \f{Exp}$, $\alpha \in
 \f{At}$, and $p \in \Sigma$. Then,
  $$\D(\f{GS}^{\Gamma}(e)) \eql  \f{GS}^{\Gamma}(\del^{\Gamma}(e)).$$
\end{proposition}
\begin{proof}
  The proof is obtained by induction on the structure of $e$. We only
  show the case $e = p$, since the other cases are similar to those in
  the proof of Proposition~\ref{prop:partialderiv}. If $\alpha \not
  \in \f{At}^{\Gamma}$, then $\f{GS}^{\Gamma}(p)=\emptyset =
  \D(\f{GS}^{\Gamma}(p))$. Also, $\del^{\Gamma}(p)=\emptyset =
  \f{GS}^{\Gamma}(\del^{\Gamma}(p))$.  Otherwise, if $\alpha \in
  \f{At}^{\Gamma}$, then $\f{GS}^{\Gamma}(p) = \{\alpha p \beta \mid
  \alpha, \beta \in \f{At}^{\Gamma} \land \forall_{bp\overline{b'}=0
    \;\in \; \Gamma} \ ( \alpha \leq b \to \beta \leq b')\}$, thus
  $\D(\f{GS}^{\Gamma}(p)) =\{\beta \in \f{At}^{\Gamma} \mid \beta \leq
  b' \text{ for all } bp\overline{b'}=0 \;\in \; \Gamma \text{ such
    that } \alpha \leq b\}$. On the other hand, $\del^{\Gamma} (p)=
  \{\Pi b' \mid bp\overline{b'}=0 \in \Gamma \land \alpha \leq b\}$.
  Thus, $\f{GS}^{\Gamma}(\del^{\Gamma}(p)) = \f{GS}^{\Gamma}(c)$,
  where $c=\prod_{bp\overline{b'}=0 \;\in \; \Gamma, \alpha \leq b }
  \;b'$. We conclude that $\f{GS}^{\Gamma}(c) = \{\beta
  \in\f{At}^{\Gamma} \mid \beta \leq b'\text{ for all }
  bp\overline{b'}=0 \;\in \; \Gamma \text{ such that }\alpha \leq b
  \}$.
\end{proof}

\vspace{-0.5cm}

\subsection{Testing Equivalence Modulo a Set of Assumptions}
\label{sec:testinggamma}
The decision procedure for testing equivalence presented before can be
easily adapted. Given a set of assumptions $\Gamma$, the set
$\f{At}^{\Gamma}$ is obtained by filtering in $\f{At}$ all atoms that
satisfy $c$ but do not satisfy $c'$, for all $c\leq c'\in \Gamma$.
The function $\f{f}$ has  to account for the new definition of
$\del^\Gamma$.

We compared this new algorithm, $\f{equiv}^\Gamma$, with $\f{equiv}$
when deciding the PCA presented in Subsection~\ref{subsec:ex}. First,
we constructed expressions $r$ and $u$ from $\Gamma$, as described
above and proved the equivalence of expressions
$t_0p_1t_1p_2t_2(t_3t_2p_3t_4p_4)\uks \overline{t_3}\overline{t_5} +
uru$ and $0 + uru$, with function $\f{equiv}$. In this case $|H|=17$.
In other words, $\f{equiv}$ needed to derive $17$ pairs of expressions
in order to reach a conclusion about the correction of program $P$.
Then, we applied function $\f{equiv}^\Gamma$ directly to the pair
$(t_0p_1t_1p_2t_2(t_3t_2p_3t_4p_4)\uks
\overline{t_3}\overline{t_5},0)$ and $\Gamma$. In this case, $|H|=5$.
Other tests, that we ran, produced similar results, but at this point
we have not carried out a study thorough enough to compare both methods.

\vspace{-0.5cm}
\section{Conclusion}
\label{sec:conclusions}
Considering the algebraic properties of KAT expressions (or even KA
expressions) it seems possible to improve the decision procedure for
equivalence. The procedure essentially computes a bisimulation (or
fails to do that if the expressions are inequivalent); thus it would be
interesting to know if, for instance the maximum bisimulation can be obtained.
Having a method that reduces the amount of used atoms, or
alternatively to resort to an external SAT solver, would
also turn the use of KAT expressions in formal
verification more feasible. Concerning Hoare logic, it would be
interesting to treat the
assignment rule within a decidable first-order theory and to
integrate the KAT decision procedure in an SMT solver.

\vspace{-0.5cm}

\bibliographystyle{eptcs}
\bibliography{kat}
\end{document}